\documentclass[preliminary,copyright,creativecommons]{eptcs}
\providecommand{\event}{INFINITY 2010} 
\usepackage{breakurl}             

\title{On Selective Unboundedness of VASS}
\author{St\'ephane Demri
\institute{LSV, CNRS, ENS de Cachan, INRIA, France}
}
\def\titlerunning{On Selective Unboundedness of VASS}
\def\authorrunning{S. Demri}

\usepackage[english]{babel}
\usepackage[latin1]{inputenc}
\usepackage[T1]{fontenc}
\usepackage{epsfig}
\usepackage{graphicx}
\usepackage{amssymb}
\usepackage{amsfonts}
\usepackage{amsmath}
\usepackage{amsthm}
\usepackage{verbatim}
\usepackage{url}
\usepackage[usenames]{color}
\usepackage{gastex}
\usepackage{floatflt}

\usepackage{stmaryrd}

\usepackage{algorithmic}
\usepackage{algorithm}

\usepackage[all]{xypic}
\usepackage{epic,eepic}
\usepackage{times}
\makeatletter
\let\c@definition\c@theorem
\let\c@lemma\c@theorem
\let\c@corollary\c@theorem
\let\c@proposition\c@theorem
\makeatother

\newif\iflong
\longfalse

\iflong
\institute{LSV, ENS Cachan, CNRS, INRIA Saclay IdF, France}
\else
\fi

\newcommand{\styletheo}[1]{\sf #1}

\newcommand{\set}[1]{\{ #1 \}}

\newcommand{\pair}[2]{( #1,#2 )}
\newcommand{\triple}[3]{( #1,#2,#3 )}

\newcommand{\tuple}[2]{( #1,\ldots,#2 )}
\newcommand{\Nat}{\mathbb{N}}
\newcommand{\Zed}{\mathbb{Z}}

\newcommand{\card}[1]{{\rm card}(#1)}

\newcommand{\aset}{X}

\newcommand{\asetter}{Z}

\newcommand{\length}[1]{|#1|}

\newcommand{\amap}{f}

\newcommand{\egdef}{\stackrel{\mbox{\begin{tiny}def\end{tiny}}}{=}} 
\newcommand{\equivdef}{\stackrel{\mbox{\begin{tiny}def\end{tiny}}}{\equivaut}} 
\newcommand{\equivaut}{\;\Leftrightarrow\;}

\newcommand{\vect}[1]{\vec{#1}}

\newcommand{\apath}{\pi}

\newcommand{\step}[1]{\xrightarrow{\!\!#1\!\!}}

\newcommand{\alocation}{q}

\newcommand{\locations}{Q}
\newcommand{\arun}{\rho}

\newcommand{\pspace}{\textsc{PSpace}}

\newcommand{\expspace}{\textsc{ExpSpace}}

\newcommand{\defstyle}[1]{\emph{#1}}

\newcommand{\pb}[1]{\textsc{#1}}

\newcommand{\atransition}{t}

\newcommand{\incr}{{\rm INC}}
\newcommand{\decr}{{\rm DEC}}
\newcommand{\amode}{\vect{mode}}

\newcommand{\avass}{\mathcal{V}}
\newcommand{\interval}[2]{[#1,#2]}
\newcommand{\avas}{\mathcal{T}}
\newcommand{\pic}[1]{{\rm maxneg}(#1)}
\newcommand{\absmaximum}[1]{{\rm scale}(#1)}
\newcommand{\maximum}[1]{{\rm max}(#1)}
\newcommand{\abs}[1]{|#1|}
\newcommand{\apseudorun}{\arun}
\newcommand{\aconstant}{{\tt C}}

\newcommand{\ainterval}{\mathcal{I}}

\newcounter{openproblem}
 \def\theopenproblem{\arabic{openproblem}}

 \newcounter{definition}[section]
 \def\thedefinition{\thesection.\arabic{definition}}

 \newcounter{theoreme}[section]
 \def\thetheoreme{\thesection.\arabic{theoreme}}

                            {\samepage\hfill\mbox{\ $\bigcirc$}\endtrivlist\par}

 \newenvironment{definition}{\refstepcounter{definition}\trivlist\item[\hskip\labelsep{\styletheo{Definition} \thedefinition}.]}%
                            {\samepage\hfill\mbox{\
 $\nabla$}\endtrivlist\par}
 \newenvironment{lemma}{\refstepcounter{theoreme}\trivlist\item[\hskip\labelsep{\styletheo{Lemma}
 \thetheoreme}.]}%
                       {\endtrivlist\par}
 \newenvironment{corollary}{\refstepcounter{theoreme}\trivlist\item[\hskip\labelsep{\styletheo{Corollary}
 \thetheoreme}.]}%
                            {\endtrivlist\par}

                           {\endtrivlist\par}

 \newenvironment{theorem}{\refstepcounter{theoreme}\trivlist\item[\hskip\labelsep{\styletheo{Theorem} \thetheoreme}.]}%
                            {\endtrivlist\par}

\begin{document}
\maketitle
%
%
%
%
\begin{abstract} 
Numerous properties of vector addition systems with states am\-ount to checking 
the (un)boundedness of some
selective feature (e.g., number of reversals, run length). 
Some of these features can be checked in exponential space by using
Rackoff's proof or its variants, combined with Savitch's theorem.
However, the question is still open for many others, e.g., reversal-boundedness.
In the paper, we introduce the class of generalized unboundedness properties 
that can be verified in exponential space
by extending Rackoff's technique, sometimes in an unorthodox way. 
We obtain new 
optimal upper bounds, 
for example for place boundedness problem, reversal-boundedness detection 
(several variants exist), strong promptness detection problem and regularity detection.
Our analysis is sufficiently refined so as we  also obtain a polynomial-space  bound when the 
dimension is fixed. 
\end{abstract}
\section{Introduction}
\label{section-introduction}
\iflong
\else {\bf Reversal-boundedness.}
\fi
A standard approach to circumvent the undecidability of the 
 reachability problem for counter 
\iflong
\else automata~\cite{Minsky67} 
\fi consists in designing
subclasses with simpler decision problems. For instance, the reachability problem is decidable for
vector addition systems with 
\iflong
\else  states (VASS)~\cite{Karp&Miller69}
\fi
or for lossy counter automata~\cite{Abdulla&Jonsson96}. 
Among the other interesting subclasses of counter automata, 
reversal-bounded counter automata verify that any counter has a boun\-ded 
number of reversals, alternations
between a nonincreasing mode and a nondecreasing mode, and vice versa. 
Reversal-boundedness remains a standard concept that 
was initially introduced in~\cite{Baker&Book74} for multistack automata.  
 A major property of such operational models is that reachability sets are effectively
definable in Presburger arithmetic~\cite{Ibarra78}, which allows decision procedures
for LTL existential model-checking and other related problems, see e.g.~\cite{Dang&Ibarra&SanPietro01}. 
However, many natural problems related to verification remain undecidable for reversal-bounded
counter automata, 
\iflong
\else
see e.g.~\cite{Dang&Ibarra&SanPietro01,Demri&Sangnier10}, 
\fi
and the 
class of reversal-bounded counter automata is not recursive~\cite{Ibarra78}. 
A significant breakthrough was achieved in~\cite{Finkel&Sangnier08} by 
designing a procedure to determine when a VASS is reversal-bounded (or weakly reversal-bounded as defined
later), 
even though the decision procedure can be nonprimitive recursive in the
worst-case. 
This means that reversal-bounded VASS 
can benefit from the known techniques for Presburger arithmetic
in order to solve their verification problems. \iflong \else \\ \fi 
\iflong
\else \noindent {\bf Selective  unboundedness.}
\fi
In order to characterize the complexity of detecting reversal-boundedness on VASS (the initial
motivation for this work), we  make a 
detour to selective
unboundedness, as explained below. 
Numerous properties of vector addition systems with states 
amounts to checking 
the (un)boundedness of some
selective feature. Some of these features can be verified in exponential space by using
Rackoff's proof or its variants, whereas the question is still open for many of 
them. In the paper, we advocate that many 
properties can be
decided as soon as we are able to decide selective unboundedness, which is a generalization
of place unboundedness for Petri nets (known to be equivalent to VASS). 
The boundedness problem was first considered in~\cite{Karp&Miller69} 
and shown decidable by simply inspecting Karp and Miller trees: the presence of
the infinity value $\infty$ (also denoted by $\omega$) is equivalent to unboundedness. 
So, 
unboundedness is equivalent to the existence of a witness 
run of the form $\vect{x_0} \step{*} \vect{x_1} \step{\apath} \vect{x_2}$ 
such that 
\iflong
\else
$\vect{x_1} \prec \vect{x_2}$ ( $\prec$ is the standard strict ordering on tuples of natural numbers).
\fi 
  In~\cite{Rackoff78}, it is shown that if there is such a run, there
is one of length at most doubly exponential. This leads to the \expspace-completeness of the
boundedness problem for VASS using the lower bound from~\cite{Lipton76} and 
\iflong
\else
Savitch's theorem.
\fi
A variant problem consists in checking whether the $i$th component is bounded, i.e., 
is there  a bound $B$ such that for every configuration
reachable from $\vect{x}_0$, its $i$th component is bounded by $B$? 
Again, inspecting Karp and Miller trees reveals the answer: 
 the presence of
the infinity value $\infty$ at the $i$th position of some extended configuration is equivalent
to $i$-unboundedness. Surprisingly, the literature often mentions this alternative problem, see e.g.~\cite{Reisig&Rozenberg98}, but never
specifies its complexity: \expspace-hardness can be obtained 
from~\cite{Lipton76}
but as far as we know, no elementary complexity upper bound has been shown. 
It might be explained by the fact that, if a VASS is unbounded, then there is a witness 
infinite run with an infinite number of distinct
configurations.  By contrast, 
it may happen that a VASS is $i$-unbounded but no infinite run has an infinite amount
of distinct values at the $i$th position of the configurations of the run. 
In the paper, we present a generalization of  place unboundedness 
by checking whether a set of components is simultaneously unbounded, possibly with some 
ordering (see  Section~\ref{section-helpful-generalization}).
This amounts to specifying in the Karp and Miller trees, the ordering with which the  value 
$\infty$ appears in the different components. \iflong \else \\ \fi
\iflong
\else {\bf Our contribution.}
\fi 
In the paper we show the following results.
\begin{enumerate}
\itemsep 0 cm
\item Detecting whether a VASS is reversal-boundedness in the sense of~\cite{Ibarra78} 
      or~\cite{Finkel&Sangnier08}  is \expspace-complete by refining the decidability results
      from~\cite{Finkel&Sangnier08} (see Theorem~\ref{theorem-RB}). 
\item We introduce the generalized unboundedness problem in 
      which many problems can be captured such as the 
      reversal-boundedness detection problems,
      the place boundedness problem, termination, 
      strong promptness detection problem, regularity detection
      and many other decision problems on VASS.
      We show that this problem can be solved in exponential space by adapting~\cite{Rackoff78} 
      even though it does not fall into the class of increasing path formulae
      introduced in~\cite{Faouzi&Habermehl09} (see Theorem~\ref{theorem-gene}).
\item Consequently, we show that regularity and strong promptness
      detection problems for VASS are in \expspace. The \expspace \ upper bound has been
      left open in~\cite{Faouzi&Habermehl09}.  Even though most of our results essentially 
      rest on the
fact that place boundedness can be solved in \expspace, our slight generalization is introduced
to obtain new complexity
upper bound for other related problems.  
\item As a by-product of our analysis and following 
      a parameterized analysis initiated in~\cite{Rosier&Yen86,Howell&Rosier87}, 
      for all the above-mentionned problems, we show that fixing the dimension of the VASS allows
to get a \pspace \ upper bound.
\end{enumerate}
\iflong
\else
\fi
The paper has also original contributions as far as  proof techniques are concerned.
First, simultaneous  unboundedness has a simple characterization in terms of
Karp and Miller trees, but we provide in the paper a witness run characterization, which allows us
to provide a complexity analysis along the lines of~\cite{Rackoff78}. 
We also provide a witness pseudo-run characterization
in which we  sometimes admit negative component values. 
This happens to be the right approach when a characterization from 
coverability graphs~\cite{Karp&Miller69,Valk&VidalNaquet81} already exists.
Apart from this unorthodox
adaptation of~\cite{Rackoff78},
 in the counterpart of Rackoff's proof
about the induction on the dimension, we provide an induction on the dimension and on the length
of the properties to be verified (see Lemma~\ref{lemma-main-induction}).
This is a genuine breakthrough compared 
to~\cite{Rackoff78,Rosier&Yen86,Habermehl97,Faouzi&Habermehl09}. 
We believe this approach is still subject to 
extensions. \\
\section{Preliminaries}

In this section, we recall the main definitions for vector addition systems with states (VASS),
without states (VAS) as well as the notions of reversal-boundedness introduced 
in~\cite{Ibarra78,Finkel&Sangnier08}. We also present the simultaneous unboundedness problem,
which slightly generalizes place unboundedness problem for Petri nets. 
 First, we write $\Nat$ [resp. $\Zed$] for the set of natural numbers [resp. integers]
and $\interval{m}{m'}$ with $m,m' \in \Zed$ to denote the set $\set{j \in \Zed: m \leq j \leq m'}$.
For $\vect{x} \in \Zed^n$,  we write
$\vect{x}(1)$, \ldots, $\vect{x}(n)$ for the entries of $\vect{x}$.
For $\vect{x}, \vect{y} \in \Zed^n$, $\vect{x} \preceq \vect{y}$ $\equivdef$
for $i \in \interval{1}{n}$, we have $\vect{x}(i) \leq \vect{y}(i)$.
We also write $\vect{x} \prec \vect{y}$ when  $\vect{x} \preceq \vect{y}$ and
 $\vect{x} \neq \vect{y}$. 

\subsection{Simultaneous unboundedness problem for VASS}
\iflong
\else {\bf VASS.}
\fi 
A \defstyle{vector addition system with states}~\cite{Hopcroft&Pansiot79} (VASS for short)
is  a finite-state automaton with transitions labelled
      by tuples of integers viewed as update functions. A \defstyle{VASS} is a structure $\avass = 
 \triple{\locations}{n}{\delta}$  such that  $\locations$ is a nonempty finite set of \defstyle{control states},
 $n \geq 1$ is the \defstyle{dimension}, and 
$\delta$ is the \defstyle{transition relation} defined as a finite set of triples in $\locations \times \Zed^n \times \locations$.  
Elements $\atransition  = \triple{\alocation}{\vect{b}}{\alocation'} \in \delta$ are called
\defstyle{transitions} and are often represented by $\alocation \step{\vect{b}} \alocation'$.
VASS with a unique control state are called \defstyle{vector addition systems} 
(VAS for short)~\cite{Karp&Miller69}.
In the sequel, a VAS $\avas$ is represented by a finite nonempty subset 
of $\Zed^n$, encoding naturally the transitions. VASS and VAS are  equivalent 
to Petri nets, see e.g.~\cite{Reutenauer90}. 
In this paper, the decision problems are defined with the VASS  model
and the decision procedures are designed for VAS, assuming that we know how the problems can be 
reduced, see e.g.~\cite{Hopcroft&Pansiot79}.
Indeed, we prefer to define problems with the help of the VASS model since 
when infinite-state transition 
systems arise in the  modeling of computational processes, there is often a natural factoring of
each system state into a control component and a memory component,
where the set of control states is typically finite. \iflong \else \\ \fi 
\iflong
\else \noindent {\bf Runs.}
\fi 
A \defstyle{configuration} of $\avass$ is defined as a pair $\pair{\alocation}{\vect{x}} \in 
\locations \times \Nat^n$ (for VAS, we simply omit the control state). 
An \defstyle{initialized VASS} is a pair of a VASS and a configuration. 
Given two configurations $\pair{\alocation}{\vect{x}}$,  $\pair{\alocation'}{\vect{x'}}$ 
and a transition $\atransition = \alocation \step{\vect{b}} \alocation'$, we 
write $\pair{\alocation}{\vect{x}} \step{\atransition} \pair{\alocation'}{\vect{x'}}$
whenever $\vect{x'} = \vect{x} + \vect{b}$.
We also write $\pair{\alocation}{\vect{x}} \step{} \pair{\alocation'}{\vect{x'}}$
when there is no need to specify the transition $\atransition$. 
The operational semantics of
VASS updates configurations, runs of such systems are essentially sequences of configurations.
Every VASS  induces 
a (possibly infinite) directed graph of configurations. Indeed, all the interesting problems on VASS
can be formulated on its \defstyle{transition system} $\pair{\locations \times \Nat^n}{\step{}}$. 
Given a VASS $\avass = \triple{\locations}{n}{\delta}$, 
a \defstyle{run} $\arun$ is a nonempty (possibly infinite) sequence 
$\arun = \pair{\alocation_0}{\vect{x_0}}, \ldots, \pair{\alocation_k}{\vect{x_k}}, \ldots$ 
of configurations 
such that $\pair{\alocation_i}{\vect{x_i}} \step{} \pair{\alocation_{i+1}}{\vect{x_{i+1}}}$ for all $i$. 
We set $Reach(\avass, \pair{\alocation_0}{\vect{x_0}}) \egdef
\set{\pair{\alocation_k}{\vect{x_k}}:  
     {\rm there \ is \ a \ finite \ run} \ 
     \pair{\alocation_0}{\vect{x_0}}, \ldots, \pair{\alocation_k}{\vect{x_k}}}$. 
A \defstyle{path} $\apath$ is a finite sequence of transitions whose successive control states respect
$\delta$ (actually this notion is mainly used for VAS without control states).
A \defstyle{pseudo-configuration} is defined as an element of $\locations \times \Zed^n$. 
When  $\apath = \atransition_1 \ldots \atransition_k$ is a path, the \defstyle{pseudo-run} 
$\apseudorun = \pair{\apath}{\pair{\alocation}{\vect{x}}}$ is defined as the sequence of pseudo-configurations
$\pair{\alocation_0}{\vect{x}_0} \cdots \pair{\alocation_k}{\vect{x}_k}$ such that
 $\pair{\alocation_0}{\vect{x}_0} = \pair{\alocation}{\vect{x}}$, and
for $i \in \interval{1}{k}$, there is 
$\atransition = \alocation_{i} \step{\vect{b}} \alocation_{i+1}$
such that
$\vect{x}_i = \vect{x}_{i-1} + \vect{b}$. 
So, we deliberately distinguish the notion of path (sequence of transitions)
from the notion of pseudo-run (sequence of elements in $\locations \times \Zed^n$ respecting
the transition from $\avass$). 
We also use the notation $\pair{\alocation}{\vect{x}} \step{\atransition} \pair{\alocation'}{\vect{x'}}$
with pseudo-configurations. 
Given a VASS $\avass$ [resp. 
a pseudo-configuration $\pair{\alocation}{\vect{x}}$, etc.] of dimension $n$,
we write $\avass(I)$ [resp. $\pair{\alocation}{\vect{x}}(I)$, etc.]
to denote the restriction of $\avass$ [resp.  $\pair{\alocation}{\vect{x}}$,  etc.]
to the components in $I \subseteq \interval{1}{n}$. \\
\iflong
\else {\bf Sizes.}
\fi 
Given $\vect{x} \in \Zed^n$, we write $\pic{\vect{x}}$ 
to denote the value 
$\maximum{\set{\maximum{0, - \vect{x}(i)}: i \in \interval{1}{n}}}$.
By extension, we write $\pic{\avass}$ to denote  
${\rm max} \set{\pic{\vect{b}}: \alocation \step{\vect{b}} \alocation'  \in \delta}$. 
Furthermore, we write $\absmaximum{\avass}$ to denote the value
$
\maximum{\set{\abs{\vect{b}(i)}: \alocation \step{\vect{b}} \alocation'  \in \delta, \
 i \in \interval{1}{n}}}
$.
For instance $\pic{\pair{-2}{3}} = 2$ and $\absmaximum{\set{\pair{-2}{3}}} = 3$. 
Given a VASS $\avass = \triple{\locations}{n}{\delta}$, we write $\length{\avass}$ to denote its
size defined by
\iflong
\else
$
\card{\locations} + 
n \times \card{\delta} \times
 (2 \times \card{\locations} + (2 + \left \lceil log_2(1 + \absmaximum{\avass}) \right\rceil))
$.
\fi
Observe that $2 + \left \lceil log_2(1 + a) \right\rceil$ is a sufficient number of bits
      to encode integers in $\interval{-a}{a}$ for $a > 0$.
Moreover $\absmaximum{\avass} \geq \pic{\avass}$, 
$\absmaximum{\avass} \leq 2^{\length{\avass}}$ and $\length{\avass} \geq 2$. \\ 
\iflong
\else \noindent {\bf Simultaneous unboundedness problem.}
\fi
Let $\pair{\avass}{\pair{\alocation_0}{\vect{x}_0}}$ be an initialized VASS of dimension $n$ 
and $\aset \subseteq \interval{1}{n}$. We say that 
$\pair{\avass}{\pair{\alocation_0}{\vect{x}_0}}$
is \defstyle{simultaneously $\aset$-unbounded} if for any
$B \geq 0$, there is a run from  $\pair{\alocation_0}{\vect{x_0}}$ to  $\pair{\alocation}{\vect{y}}$ 
such that
for $i \in \aset$, we have $\vect{y}(i) \geq B$. 
When $\aset = \set{j}$, we say that $\pair{\avass}{\pair{\alocation_0}{\vect{x}_0}}$
 is \defstyle{$j$-unbounded}. 
It is clear that  $\pair{\avass}{\pair{\alocation_0}{\vect{x}_0}}$ 
is bounded (i.e., $Reach(\avass,\pair{\alocation_0}{\vect{x}_0})$ is finite) iff for all $j$, 
$\pair{\avass}{\pair{\alocation_0}{\vect{x}_0}}$  is not $j$-unbounded. 
\iflong
\else The \defstyle{simultaneous unboundedness problem} is defined as follows: 
given an initialized VASS 
$\pair{\avass}{\pair{\alocation_0}{\vect{x}_0}}$ of dimension $n$ and $\aset 
\subseteq \interval{1}{n}$, is $\pair{\avass}{\pair{\alocation_0}{\vect{x}_0}}$ 
simultaneously $\aset$-un\-boun\-ded?
\fi
\begin{theorem} \cite{Karp&Miller69}
Simultaneous unboundedness problem is decidable.
\end{theorem}
This follows from~\cite{Karp&Miller69,Valk&VidalNaquet81}:
$\pair{\avass}{\pair{\alocation_0}{\vect{x}_0}}$ 
is simultaneously $\aset$-unbounded iff
the coverability graph 
$CG(\avass,\pair{\alocation_0}{\vect{x}_0})$ (see e.g.,~\cite{Karp&Miller69,Valk&VidalNaquet81}) contains an extended configuration $\pair{\alocation}{\vect{y}}$ 
such that
$\vect{y}(\aset) = \vect{\infty}$ (for $\alpha \in \Zed \cup \set{\infty}$, we write
$\vect{\alpha}$ to denote any vector of dimension $n \geq 1$ whose component values
are $\alpha$). 
\subsection{Standard reversal-boundedness and its new variant}
A \defstyle{reversal} for a counter occurs in a run 
when there is an alternation from nonincreasing mode to nondecreasing mode and vice-versa.
A VASS is \defstyle{reversal-bounded} whenever 
there is $r \geq 0$ such that
for any run, every counter makes no more than $r$ reversals.
This class of VASS  has been introduced and studied in~\cite{Ibarra78}, partly inspired
by similar restrictions on multistack automata~\cite{Baker&Book74}. 
In spite of the fact
that the problem of deciding  whether a counter automaton (VASS with zero-tests) is reversal-bounded 
is undecidable~\cite{Ibarra78}, we shall see that reversal-bounded counter automata have 
numerous fundamental properties.
Moreover, a breakthrough has been achieved
in~\cite{Finkel&Sangnier08} by establishing that checking whether a VASS is reversal-bounded is decidable.
The decidability proof in~\cite{Finkel&Sangnier08} provides a decision procedure that requires
nonprimitive recursive time in the worst-case since Karp and Miller trees need to be 
built~\cite{Karp&Miller69,Valk&VidalNaquet81}. 
\iflong  \fi
Let $\avass = \triple{\locations}{n}{\delta}$ be a VASS.
Let us define the auxiliary VASS $\avass_{rb}  = \triple{\locations'}{2n}{\delta'}$
such that essentially, 
the $n$ new components in $\avass_{rb}$ count the number of reversals for each component
from $\avass$. 
We set
$\locations' = \locations \times  \set{\decr,\incr}^n$ and,
for each $\vect{v} \in \set{\decr,\incr}^n$ and $i \in \interval{1}{n}$, $\vect{v}(i)$ 
encodes whether component $i$ is in a decreasing mode or in an increasing mode. 
Moreover,
$\pair{\alocation}{\amode} \step{\vect{b'}} \pair{\alocation'}{\amode'}  \in \delta'$ 
(with $\vect{b'} \in \Zed^{2n}$)  $\equivdef$
      there is $\alocation \step{\vect{b}} \alocation' \in \delta$ such that
     $\vect{b'}(\interval{1}{n}) = \vect{b}$ and 
for every $i \in \interval{1}{n}$, one of the conditions below is satisfied:
\begin{itemize}
\itemsep 0 cm
\item $\vect{b}(i) < 0$, $\amode(i) = \amode'(i) = \decr$ and $\vect{b'}(n+i) = 0$,
\item $\vect{b}(i) < 0$, $\amode(i) = \incr$, $\amode'(i) = \decr$ and $\vect{b'}(n+i) = 1$,
\item $\vect{b}(i) > 0$, $\amode(i) = \amode'(i) = \incr$ and $\vect{b'}(n+i) = 0$,
\item $\vect{b}(i) > 0$, $\amode(i) = \decr$, $\amode'(i) = \incr$ and $\vect{b'}(n+i) = 1$,
\item $\vect{b}(i) = 0$, $\amode(i) = \amode'(i)$ and $\vect{b'}(n+i) = 0$. 
\end{itemize} 
Initialized VASS $\pair{\avass}{\pair{\alocation}{\vect{x}}}$ is 
\defstyle{reversal-bounded}~\cite{Ibarra78}
$\equivdef$ for  $i \in \interval{n+1}{2n}$, $\set{\vect{y}(i): \ \exists \ {\rm run} \ \pair{\alocation_{rb}}{\vect{x}_{rb}} \step{*}  
\pair{\alocation'}{\vect{y}} \ {\rm in } \ \avass_{rb}}$ 
is finite with $\alocation_{rb} = \pair{\alocation}{\vect{\incr}}$, 
$\vect{x}_{rb}$ restricted to the $n$ first components is $\vect{x}$
and $\vect{x}_{rb}$ restricted to the $n$ last components is $\vect{0}$. 
When $r \geq \maximum{\set{\vect{y}(i): \ \exists \ {\rm run} \ \pair{\alocation_{rb}}{\vect{x}_{rb}} \step{*}  
\pair{\alocation'}{\vect{y}} \ {\rm in } \ \avass_{rb}}: i 
\in \interval{n+1}{2n}}$  $\pair{\avass}{\pair{\alocation}{\vect{x}}}$ is said to 
be \defstyle{$r$-reversal-bounded}.
For $i \in \interval{1}{n}$, when 
$\set{\vect{y}(n+i): \ \exists \ {\rm run} \ \pair{\alocation_{rb}}{\vect{x}_{rb}} \step{*}  
\pair{\alocation'}{\vect{y}} \ {\rm in } \ \avass_{rb}}$ 
is finite, we say that $\pair{\avass}{\pair{\alocation}{\vect{x}}}$ is 
\defstyle{reversal-bounded with respect to $i$}. 
\iflong
\else
A VASS $\avass$ is \defstyle{ globally reversal-bounded} iff  there is $r \geq 0$ such that
for every configuration $\pair{\alocation}{\vect{x}}$, 
 $\pair{\avass}{\pair{\alocation}{\vect{x}}}$ is 
$r$-reversal-bounded. Global reversal-boundedness detection can be easily reduced
to  reversal-boundedness detection. 
Indeed, it is sufficient to introduce a new control state $\alocation_{new}$ that contains as many self-loops
as the dimension $n$ and each self-loop $i$ increments the $i$th component. Then, nondeterministically we jump
to the rest of the VASS. 
In this way, $\pair{\avass'}{\pair{\alocation_{new}}{\vect{0}}}$ is 
reversal-bounded ($\avass'$ is the new VASS obtained as a variant of $\avass$) iff
 $\avass$ is globally reversal-bounded (forthcoming upper bounds will apply to this
problem too). 
\fi 

Reversal-boundedness for counter automata, and {\em a fortiori} for VASS, is very
appealing because reachability sets
are semilinear as recalled below.

\begin{theorem} \label{theorem-ibarra} \cite{Ibarra78} Let $\pair{\avass}{\pair{\alocation}{\vect{x}}}$
be an $r$-reversal-bounded VASS. For each control state $\alocation'$,
the set $\set{\vect{y} \in \Nat^n: \ \exists \ {\rm run} \ 
\pair{\alocation}{\vect{x}} \step{*}  \pair{\alocation'}{\vect{y}}}$
is effectively semilinear. 
\end{theorem}

This means that one can compute effectively a Presburger formula that characterizes
precisely the reachable configurations whose control state is $\alocation'$.
So, detecting reversal-boundedness for VASS, which can be easily reformulated as an unboundedness problem,
is worth the effort since semilinearity follows and then decision procedures for Presburger arithmetic can be used. 
\begin{lemma} \label{lemma-reversalbounded-unboundedness}
$\pair{\avass}{\pair{\alocation}{\vect{x}}}$  
is reversal-bounded with respect to $i$  iff
$\pair{\avass_{rb}}{\pair{\alocation_{rb}}{\vect{x}_{rb}}}$ 
 is not  $(n+i)$-unbounded. 
\end{lemma}
An interesting extension of 
reversal-boundedness is introduced in~\cite{Finkel&Sangnier08,Sangnier08} 
for which we only count the number of reversals when they occur for a counter value above a given bound $B$.
For instance, finiteness of the reachability set implies reversal-boundedness in 
the sense of~\cite{Finkel&Sangnier08,Sangnier08}, which we shall call \defstyle{weak reversal-boundedness}.
Let $\avass = \triple{\locations}{n}{\delta}$ be a VASS and a bound 
$B \in \Nat$.
Instead of defining a counter automaton $\avass_{rb}$ as done to characterize (standard) reversal-boundedness, we define
directly an infinite directed graph that corresponds to a variant of the transition system of $\avass_{rb}$: still, there
are $n$ new counters that record the number of reversals but only if they occur above a bound $B$. That is why, 
the  infinite directed graph
$TS_B = \pair{\locations \times  \set{\decr,\incr}^n \times \Nat^{2n}}{\step{}_B}$
is defined as follows:
$\pair{\alocation,\amode}{\vect{x}}  \step{}_B \pair{\alocation',\amode'}{\vect{x}'}$ $\equivdef$ there 
is a transition
$\alocation \step{\vect{b}} \alocation' \in \delta$ such that $\vect{x}'(\interval{1}{n}) = \vect{x}(\interval{1}{n}) + \vect{b}$,
 and for every $i \in \interval{1}{n}$, one of the conditions below is satisfied:
\begin{itemize}
\itemsep 0 cm
\item $\vect{b}(i) < 0$, $\amode(i) = \amode'(i) = \decr$ and $\vect{b'}(n+i) = 0$,
\item $\vect{b}(i) < 0$, $\amode(i) = \incr$, $\amode'(i) = \decr$, $\vect{x}(i) \leq B$ and $\vect{b'}(n+i) = 0$,
\item $\vect{b}(i) < 0$, $\amode(i) = \incr$, $\amode'(i) = \decr$, $\vect{x}(i) > B$ and $\vect{b'}(n+i) = 1$,
\item $\vect{b}(i) > 0$, $\amode(i) = \amode'(i) = \incr$ and $\vect{b'}(n+i) = 0$,
\item $\vect{b}(i) > 0$, $\amode(i) = \decr$, $\amode'(i) = \incr$, $\vect{x}(i) > B$ and $\vect{b'}(n+i) = 1$,
\item $\vect{b}(i) > 0$, $\amode(i) = \decr$, $\amode'(i) = \incr$, $\vect{x}(i) \leq B$ and $\vect{b'}(n+i) = 0$,
\item $\vect{b}(i) = 0$, $\amode(i) = \amode'(i)$ and $\vect{b'}(n+i) = 0$. 
\end{itemize} 
Initialized VASS $\pair{\avass}{\pair{\alocation}{\vect{x}}}$ is 
 \defstyle{weakly reversal-bounded}~\cite{Finkel&Sangnier08}
$\equivdef$ there is some $B \geq 0$ such that for 
$i \in \interval{n+1}{2n}$, $\set{\vect{y}(i): \  \pair{\alocation_{rb}}{\vect{x}_{rb}} \step{*}_B  
\pair{\alocation'}{\vect{y}} \ {\rm in } \ TS_B}$ 
is finite. 
When $r \geq \maximum{\set{\vect{y}(i): \pair{\alocation_{rb}}{\vect{x}_{rb}} \step{*}_B  
\pair{\alocation'}{\vect{y}} \ {\rm in } \ TS_B}: i 
\in \interval{n+1}{2n}}$  $\pair{\avass}{\pair{\alocation}{\vect{x}}}$ is said to 
be \defstyle{$r$-reversal-$B$-bounded}.
Observe that whenever  $\pair{\avass}{\pair{\alocation}{\vect{x}}}$ is  $r$-reversal-bounded, 
 $\pair{\avass}{\pair{\alocation}{\vect{x}}}$ is $r$-reversal-$0$-bounded. 
\iflong
\else
As shown in~\cite{Finkel&Sangnier08},  $r$-reversal-$B$-boundedness for some known $r$ and $B$ also leads
to effective semilinearity of reachability sets and therefore detecting weak reversal-boundedness is also
worth the effort. 
\fi
\iflong
\else
The \defstyle{reversal-boundedness detection problem} is defined as follows: given an 
initialized VASS $\pair{\avass}{\pair{\alocation}{\vect{x}}}$  of dimension $n$
and $i \in \interval{1}{n}$, is $\pair{\avass}{\pair{\alocation}{\vect{x}}}$  
reversal-bounded with respect to the component $i$? We also consider  the variant with
weak reversal-boundedness.
\fi

Let us conclude this section by Lemma~\ref{lemma-rb-hp} below. 
The proof is essentially based on~\cite[Lemma 2.1]{Hopcroft&Pansiot79} and on the definition
of the initialized VASS $\pair{\avass_{rb}}{\pair{\alocation_{rb}}{\vect{x}_{rb}}}$.
The key properties are that the dimension increases only linearly and the scale ``only'' exponentially
in the dimension.

\begin{lemma} \label{lemma-rb-hp} 
Given $\avass = \triple{\locations}{n}{\delta}$ and a configuration 
$\pair{\alocation}{\vect{x}}$, one can effectively build in po\-ly\-no\-mial space an initialized VAS
$\pair{\avas}{\vect{x}'}$ of dimension $2n+3$ such that  $\pair{\avass}{\pair{\alocation}{\vect{x}}}$ is reversal-bounded
with respect to $i$ iff $\pair{\avas}{\vect{x}'}$ is  not $(n+i)$-unbounded. Moreover,
$\absmaximum{\avas} = \maximum{(\card{\locations} \times 2^n + 1)^2, \absmaximum{\avass}}$.
\end{lemma}

Note also that by using the simple reduction from VASS to VAS that increases the dimension by the number of
control states, we would increase exponentially the dimension, which would disallow us to obtain forthcoming
optimal complexity bounds. 
In Lemma~\ref{lemma-weak-reversal-boundedness}, 
we shall explain how to reduce weak reversal-boundedness detection
to a generalization of $(n+i)$-unboundedness.
\section{Generalized Unboundedness Properties}
\label{section-other-properties}

In this section, we essentially introduce the  generalized unboundedness problem
and we show how several detection problems can be naturally reduced to it. 

\subsection{Witness runs for simultaneous unboundedness}

We know that  $\pair{\avass}{\pair{\alocation_0}{\vect{x}_0}}$ is $i$-unbounded iff 
the coverability graph $CG(\avass,\pair{\alocation_0}{\vect{x_0}})$ 
(see e.g.,~\cite{Karp&Miller69,Valk&VidalNaquet81})
contains an extended configuration with  $\infty$ on the $i$th component.
This is a simple characterization whose main disadvantage is to induce 
 a nonprimitive recursive decision procedure in the worst-case. 
By contrast, un\-boun\-ded\-ness of   $\pair{\avass}{\pair{\alocation_0}{\vect{x}_0}}$
(i.e. $i$-un\-boun\-ded\-ness for some $i \in \interval{1}{n}$)
is equivalent to the existence of witness run of the form 
$\pair{\alocation_0}{\vect{x}_0} \step{*} \pair{\alocation_1}{\vect{x}_1} \step{+} 
\pair{\alocation_2}{\vect{x}_2}$ 
such that $\vect{x}_1 \prec \vect{x}_2$ and $\alocation_1 = \alocation_2$. 
In~\cite{Rackoff78}, it is shown that if there is such a run, there
is one of length at most doubly exponential. Given a component $i \in \interval{1}{n}$, a natural adaptation
to 
$i$-unboundedness is to check the existence of
a run of the form 
$\pair{\alocation_0}{\vect{x}_0} \step{*} \pair{\alocation_1}{\vect{x}_1} \step{\apath} 
\pair{\alocation_2}{\vect{x}_2}$ 
such that $\vect{x}_1 \prec \vect{x}_2$, $\alocation_1 = \alocation_2$ 
and $\vect{x}_1(i) < \vect{x}_2(i)$. By inspecting the 
proof in~\cite{Rackoff78}, one can  show that if there is such a run, then there is one of length at most
doubly exponential. However, although existence of such a run is a sufficient condition for
$i$-unboundedness (simply iterate $\apath$ infinitely), this is  not necessary
as shown on the VASS below:
$$
\begin{picture}(100,7)(-30,0)
\setlength{\unitlength}{0.3mm}
 
  \node(A)(10,2){$A$}
  \node(B)(90,2){$B$}
  
  \drawloop[loopangle=180](A){${\tiny \left(
  \begin{array}{c}
   1 \\ 0
   \end{array}
  \right)}$}  
  \drawedge(A,B){${\tiny \left(
  \begin{array}{c}
   0 \\ 0
   \end{array}
  \right)}$}
  \drawloop[loopangle=0](B){${\tiny \left(
  \begin{array}{c}
   -1 \\ 1
   \end{array}
  \right)}$}
 
\end{picture}
$$
The second component is unbounded from $\pair{A}{\vect{0}}$ 
but no run
 $\pair{A}{\vect{0}} \step{*} \pair{\alocation}{\vect{x_1}} \step{\apath} \pair{\alocation}{\vect{x_2}}$ 
with $\vect{x_1} \prec \vect{x_2}$, $\vect{x_1}(2) < \vect{x_2}(2)$ and $\alocation \in \set{A,B}$ exists.
Indeed, in order to increment the second component, the first component needs first to be incremented.
The ultimate condition for simultaneous unboundedness needs to specify
 the different ways to introduce
the value $\infty$ along a given branch of the Karp and Miller trees. 
This is done thanks to the condition 
 ${\rm PB}_{\sigma}$ defined below and generalized in Section~\ref{section-helpful-generalization}.
A \defstyle{disjointness sequence} is a nonempty sequence 
$\sigma = \aset_1 \cdot \cdots \cdot \aset_K$ of nonempty subsets of 
$\interval{1}{n}$ such that for $i \neq i'$, $\aset_i \cap \aset_{i'} = \emptyset$ (consequently $K \leq n$). 
A run of the form $$\pair{\alocation_0}{\vect{x}_0} \step{\apath_0'} \pair{\alocation_1}{\vect{x}_1} 
\step{\apath_1} \pair{\alocation_2}{\vect{x}_2}
             \step{\apath_1'} 
              \cdots 
             \step{\apath_{K-1}'} \pair{\alocation_{2K-1}}{\vect{x}_{2K-1}} \step{\apath_K} 
             \pair{\alocation_{2K}}{\vect{x}_{2K}}$$
satisfies the \defstyle{property  ${\rm PB}_{\sigma}$} (Place Boundedness with respect to a disjointness sequence $\sigma$)
iff the  conditions below hold true:
\begin{description}
\itemsep 0 cm
\item[(P0)]  For  $l \in \interval{1}{K}$, $\alocation_{2l-1} = \alocation_{2l}$. 
\item[(STRICT)] for  $l \in \interval{1}{K}$ and  $j \in \aset_l$,
   $\vect{x_{2l-1}}(j) < \vect{x_{2l}}(j)$.
\item[(NONSTRICT)] For  $l \in \interval{1}{K}$ and
             $j \in (\interval{1}{n} \setminus \aset_l)$, 
            $\vect{x_{2l}}(j) < \vect{x_{2l-1}}(j)$ implies 
            $j \in  \underset{l' \in \interval{1}{l-1}}{\bigcup} \aset_{l'}$.
\end{description}
Observe that when (STRICT) holds, the condition (NONSTRICT) is equivalent to: 
for all $l \in \interval{1}{K}$ and
            all $j \not \in   \underset{l' \in \interval{1}{l-1}}{\bigcup} \aset_{l'}$,
            we have $\vect{x}_{2l-1}(j) \leq \vect{x}_{2l}(j)$.
Consequently, for all $l \in \interval{1}{K}$ and all paths of the form $(\apath_l)^k$ for some $k \geq 1$,
the effect on the $j$th component may be negative only if
$j \in \underset{l' \in \interval{1}{l-1}}{\bigcup} \aset_{l'}$. 
It is now time to provide a witness run characterization for simultaneous $\aset$-unboundedness
that is a direct consequence of the properties of the coverability 
graphs~\cite{Valk&VidalNaquet81}. 
\begin{lemma} \label{lemma-simul}
 Let $\pair{\avass}{\pair{\alocation_0}{\vect{x}_0}}$ be an initialized VASS of dimension $n$
and $\aset \subseteq \interval{1}{n}$. Then, 
 $\pair{\avass}{\pair{\alocation_0}{\vect{x}_0}}$ 
is simultaneously $\aset$-unbounded  iff
there is a run $\arun$ starting at $\pair{\alocation_0}{\vect{x}_0}$ satisfying 
${\rm PB}_{\sigma}$
for some disjointness sequence $\sigma =  \aset_1 \cdot \cdots \cdot \aset_K$ such that
$\aset \subseteq (\aset_1 \cup \cdots \cup \aset_K)$
and $\aset \cap \aset_K \neq \emptyset$. 
\end{lemma}

Consequently, $\pair{\avass}{\pair{\alocation_0}{\vect{x}_0}}$ is $i$-unbounded iff
there is a run $\arun$ starting at $\pair{\alocation_0}{\vect{x}_0}$ satisfying 
${\rm PB}_{\sigma}$
for some disjointness sequence $\sigma =  \aset_1 \cdot \cdots \cdot \aset_K$ with
$i \in \aset_K$. 
This
can be expressed in the logical formalisms from~\cite{Yen92,Faouzi&Habermehl09}
but this requires a formula of exponential size in the dimension because an exponential number of 
disjointness sequences needs to be taken into account. 
By contrast, each disjunct has only polynomial-size in $n$. 
The
path formula 
\iflong
\else
looks like that:
\fi
$$
\bigvee_{\aset_1 \cdots \aset_K, i \in \aset_K}
\hspace{-0.2in} \exists \ \vect{x}_1, \ldots, \vect{x}_{2K}  \ 
\bigwedge_{l=1}^{K}
(\bigwedge_{j \in \aset_l} \vect{x}_{2l-1}(j) <  \vect{x}_{2l}(j)) \wedge
(\bigwedge_{j \not \in (\aset_1 \cup \cdots \cup \aset_{l-1})} 
\hspace{-0.3in} \vect{x}_{2l-1}(j) \leq  \vect{x}_{2l}(j))
$$
It is worth noting that 
the satisfaction of ${\rm PB}_{\sigma}$
does not imply  
$\vect{x}_1 \preceq \vect{x}_{2K}$. This prevents us 
from defining this condition with an increasing path formula~\cite{Faouzi&Habermehl09}
and therefore the \expspace \ upper bound established in~\cite{Faouzi&Habermehl09} 
does not apply directly to $i$-unboundedness.  
\subsection{A helpful  generalization}
\label{section-helpful-generalization}
We introduce below a slight generalization of the above properties ${\rm PB}_{\sigma}$ in order
to underline their essential features and to provide a uniform treatment. Moreover,
this will allow us to express new properties, for instance  for regularity detection.
The conditions (STRICT) and (NONSTRICT) specify inequality constraints between
component values. We  introduce intervals in place of such
constraints.
An \defstyle{interval} is an expression of one of the forms 
$]-\infty,+\infty[$,  $[a,+\infty[$, 
$]-\infty,b]$ or $[a,b]$  for some $a,b \in \Zed$ (with the obvious interpretation). 
\iflong
\else
A \defstyle{generalized unboundedness property} $\mathcal{P} = 
 \tuple{\ainterval_1}{\ainterval_K}$ is
a nonempty sequence  of $n$-tuples of intervals. 
\fi
The \defstyle{length} of $\mathcal{P}$ is $K$ and its \defstyle{scale} is equal to 
the maximum between 1 and 
the maximal absolute value of integers occurring in the interval expressions of 
$\mathcal{P}$ (if any).
A run of the form
 $\pair{\alocation_0}{\vect{x_0}} \step{\apath_0'} \pair{\alocation_1}{\vect{x}_1} \step{\apath_1} 
\pair{\alocation_2}{\vect{x}_2} 
             \step{\apath_1'} \pair{\alocation_3}{\vect{x}_3}  
             \cdots 
             \step{\apath_{K-1}'} \pair{\alocation_{2K-1}}{\vect{x}_{2K-1}} \step{\apath_K} 
             \pair{\alocation_{2K}}{\vect{x}_{2K}}$
satisfies the property $\mathcal{P}$ $\equivdef$ 
 (P0) and the conditions below hold true:
\begin{description}
\itemsep 0 cm
\item[(P1)]  For  $l \in \interval{1}{K}$ and $j \in \interval{1}{n}$,
             we have $\vect{x}_{2l}(j) - \vect{x}_{2l-1}(j) \in \ainterval_l(j)$.
\item[(P2)]  For  $l \in \interval{1}{K}$ and $j \in \interval{1}{n}$,
             if $\vect{x_{2l}}(j) - \vect{x_{2l-1}}(j) < 0$, then there is
             $l' < l$ \iflong\else s.t. \fi  $\vect{x}_{2l'}(j) - \vect{x}_{2l'-1}(j) > 0$. 
\end{description}
Given a run $\arun$, we say that it \defstyle{satisfies $\mathcal{P}$} if it admits a decomposition
satisfying the adequate conditions. 
By extension,  $\pair{\avass}{\pair{\alocation_0}{\vect{x}_0}}$ 
\defstyle{satisfies $\mathcal{P}$} $\equivdef$
there is a finite run starting at $\pair{\alocation_0}{\vect{x}_0}$ 
satisfying $\mathcal{P}$. 
It is easy to see that condition (P1) [resp. (P2)] is a quantitative counterpart
for condition (STRICT) [resp. (NONSTRICT)]. 
\iflong
\else
The \defstyle{generalized unboundedness problem} is defined as follows: 
given an initialized VASS $\pair{\avass}{\pair{\alocation_0}{\vect{x}_0}}$ and 
a generalized unboundedness property 
$\mathcal{P}$, does $\pair{\avass}{\pair{\alocation_0}{\vect{x}_0}}$ satisfy $\mathcal{P}$?
\fi 
Let us first forget about control states: 
 we can safely restrict ourselves to VAS without any loss of generality, as
it is already the case for many properties. 
\begin{lemma} \label{lemma-reduction}
 There is a logspace many-one reduction from
the generalized unboundedness problem for VASS to the generalized unboundedness problem for VAS.
Moreover, an instance $\pair{\pair{\avass}{\pair{\alocation}{\vect{x}}}}{\mathcal{P}}$
is reduced to an instance $\pair{\pair{\avas}{\vect{x}'}}{\mathcal{P}'}$ such that
(1) if $\avass$ is of dimension $n$, then $\avas$ is of dimension $n+3$,
(2) $\mathcal{P}$ and $\mathcal{P}'$ have the same length and scale
and (3) $\absmaximum{\avas} = 
\maximum{(\card{\locations} + 1)^2, 
\absmaximum{\avass}}$ where $\locations$ is the set of control states of $\avass$.
\end{lemma} 

The proof is essentially based on~\cite[Lemma 2.1]{Hopcroft&Pansiot79}.
Generalized unboundedness properties can be expressed in even more general formalisms
for which decidability is known. However, in Section~\ref{section-expspace}, we shall establish 
\expspace-completeness. 

\begin{theorem} \cite{Faouzi&Habermehl09} The generalized unboundedness problem  is decidable.
\end{theorem}

Given $\pair{\avass}{\pair{\alocation_0}{\vect{x_0}}}$, 
the existence of a run from $\pair{\alocation_0}{\vect{x_0}}$
satisfying $\mathcal{P}$ can be easily 
expressed in Yen's path logic~\cite{Yen92} and the generalized unboundedness problem is therefore 
decidable by~\cite[Theorem 3]{Faouzi&Habermehl09} and~\cite{Mayr84,Kosaraju82}. 
We cannot rely on~\cite[Theorem 3.8]{Yen92} for decidability since \cite[Lemma 3.7]{Yen92} contains a flaw,
as observed  in~\cite{Faouzi&Habermehl09}. \cite{Faouzi&Habermehl09} 
 precisely  establishes that satisfiability in Yen's path logic
is equivalent to the reachability problem for VASS.  
Moreover, it is worth noting that
the reduction from the reachability problem to 
satisfiability~\cite[Theorem 2]{Faouzi&Habermehl09} uses path formulae that cannot be expressed
as generalized unboundedness properties. 
Observe that the \expspace \ upper bound obtained 
for increasing path formulae in~\cite[Section 6]{Faouzi&Habermehl09} cannot be used
herein since obviously generalized unboundedness properties are not 
necessarily increasing.
That is why, we need
directly to extend Rackoff's proof for boundedness~\cite{Rackoff78}. 

\subsection{From regularity to reversal-boundedness detection}
\label{section-expressive-power}

In this section, we briefly explain how
simultaneous unboundedness problem, regularity detection, strong promptness detection
and weak reversal-boundness detection can be reduced to generalized unboundedness problem.
This will allow us to obtain \expspace \ upper bound for all these problems. 

\iflong
\else
\paragraph{Simultaneous unboundedness problem.} It is easy to show that 
every property ${\rm PB}_{\sigma}$ can be encoded as a generalized 
unboundedness property $\mathcal{P}_{\sigma}$ with length $K \leq n$ and 
$\absmaximum{\mathcal{P}_{\sigma}} = 1$.
Indeed, from a disjointness sequence $\sigma = \aset_1 \cdots \aset_K$, 
we define  $\mathcal{P}_{\sigma} = \tuple{\ainterval_1}{\ainterval_K}$ as follows.
For $l \in \interval{1}{K}$ and  $j \in \interval{1}{n}$, if
     $j \in \aset_l$ then $\ainterval_l(j) = [1,+\infty[$.
     Otherwise, if $j \in (\interval{1}{n} \setminus  (\bigcup_{1 \leq l' \leq l} \aset_{l'}))$,
     then  $\ainterval_l(j) = [0,+\infty[$, otherwise $\ainterval_l(j) = ]-\infty,+\infty[$.
It is then easy to check that $\sigma$ and ${\rm PB}_{\sigma}$ define the same set of runs.
\fi 

\paragraph{Regularity detection.}
Another  example of properties that can be  encoded by generalized unboundedness properties
comes from the witness run characterization for nonregularity, 
see e.g.~\cite{Valk&VidalNaquet81,Faouzi&Habermehl09}.
Nonregularity of an initialized VASS $\pair{\avass}{\pair{\alocation_0}{\vect{x_0}}}$  
is equivalent to the existence of a run of the form
$
\pair{\alocation_0}{\vect{x_0}} \step{\apath_0'}
\pair{\alocation_1}{\vect{x_1}} \step{\apath_1}
\pair{\alocation_2}{\vect{x_2}} \step{\apath_1'}
\pair{\alocation_3}{\vect{x_3}} \step{\apath_2} \pair{\alocation_4}{\vect{x_4}}
$
such that $\alocation_1 = \alocation_2$, $\alocation_3 = \alocation_4$, 
there is $i \in \interval{1}{n}$ such that
$\vect{x_1} \prec \vect{x_2}$,  $\vect{x_4}(i) < \vect{x_3}(i)$ and
for all $j \in \interval{1}{n}$ such that $\vect{x_4}(j) < \vect{x_3}(j)$,
we have $\vect{x_1}(j) < \vect{x_2}(j)$, see e.g.~\cite{Valk&VidalNaquet81,Faouzi&Habermehl09}. 
Consequently, nonregularity condition can be viewed as a disjunction of generalized 
unboundedness properties of the form
$\pair{\ainterval_1^i}{\ainterval_2^i}$ where
$\ainterval_1^i(i) = [1,+\infty[$, $\ainterval_2^i(i) = ]-\infty,-1]$,
and for $j \neq i$, we have  $\ainterval_1^i(j) = [0,+\infty[$ and 
 $\ainterval_2^i(j) = ]-\infty,+\infty[$. 

\paragraph{Strong promptness detection.}
We show below how the strong promptness detection 
problem can be reduced to the simultaneous unboundedness problem, leading to an \expspace 
\ upper bound.
\iflong
\else
 The \defstyle{strong promptness detection problem} is defined as follows~\cite{Valk&Jantzen85}: given 
\iflong an initialized VASS \fi 
$\pair{\triple{\locations}{n}{\delta}}{\pair{\alocation}{\vect{x}}}$ and a partition
              $\pair{\delta_I}{\delta_E}$ of $\delta$, 
is there $B \in \Nat$ such that for every run $\pair{\alocation}{\vect{x}} \step{*} 
\pair{\alocation'}{\vect{x}'}$, there is no run 
$\pair{\alocation'}{\vect{x}'} \step{\apath} 
\pair{\alocation''}{\vect{x}''}$ using only transitions 
from $\delta_I$ and of length more than $B$  ($\apath \in \delta^*_I$) ?
The transitions in $\delta_I$ are called \defstyle{internal} and strong promptness guarantees
that sequences of internal transitions cannot be arbitrarily long.
\fi 
Let us consider below the VASS $\avass$ of dimension 1 with $\delta_I$ made of the
two transitions in bold. 
$$
\begin{picture}(100,5)(-40,0)
\setlength{\unitlength}{0.2mm}
 
  \node(A)(10,2){$A$}
  \node(B)(90,2){$B$}
  \node(C)(170,2){$C$}
  
  \drawloop[loopangle=180](A){$+1$}  
  \drawedge(A,B){$0$}
  \drawedge[curvedepth=5,linewidth=1.7](B,C){
  $-1$
  }
  \drawedge[curvedepth=5,linewidth=1.7](C,B){
  $-1$
  }  
\end{picture}
$$
$\pair{\avass}{\pair{A}{0}}$ is not strongly prompt and
there is no $\pair{A}{0} \step{*} \pair{\alocation}{\vect{x}} \step{\apath} 
\pair{\alocation}{\vect{y}}$ for some $\alocation \in \set{A,B,C}$ such that 
$\vect{x} \preceq \vect{y}$, $\apath$ is nonempty and contains only 
transitions in $\delta_I$.
\begin{lemma} \label{lemma-promptness}
There is a logspace reduction from strong promptness detection problem to 
the complement of simultaneous unboundedness problem. 
\end{lemma}

\paragraph{Weak reversal-boundedness detection.}
Complement of weak reversal-boundedness involves two universal quantifications (on $B$ and $r$)
that can be understood as  simultaneous unboundedness properties.
Lemma~\ref{lemma-weak-reversal-boundedness} below is a key intermediate result in our investigation. 

\begin{lemma} \label{lemma-weak-reversal-boundedness}
Given a VASS $\avass = \triple{\locations}{n}{\delta}$ and a configuration 
$\pair{\alocation}{\vect{x}}$,
$\pair{\avass}{\pair{\alocation}{\vect{x}}}$ is not weakly reversal-bounded
with respect to $i$ iff $\pair{\avass_{rb}}{\pair{\alocation_{rb}}{\vect{x}_{rb}}}$
has a run  satisfying ${\rm PB}_{\sigma}$ for some 
disjointness sequence $\sigma = \aset_1 \cdots \aset_K$
with $n+i \in \aset_K$ and $i \in  (\aset_1 \cup \cdots \cup \aset_{K-1})$.
\end{lemma}
As a corollary, we are in a position to present a witness 
run characterization for weak reversal-boun\-ded\-ness detection.
$\pair{\avass}{\pair{\alocation_0}{\vect{x}_0}}$ is not weakly reversal-bounded with respect to $i$ iff
there exist a disjointness sequence $\sigma = \aset_1  \cdots \aset_K$ and a run 
$\pair{\alocation_0}{\vect{x}_0} \step{\apath_0'} \pair{\alocation_1}{\vect{x}_1} 
\step{\apath_1} \pair{\alocation_2}{\vect{x}_2}
             \step{\apath_1'} 
              \cdots 
             \step{\apath_{K}'} \pair{\alocation_{2K+1}}{\vect{x}_{2K+1}} \step{\apath_{K+1}} 
             \pair{\alocation_{2K+2}}{\vect{x}_{2K+2}}$
such that (1) $\apath_{K+1}$ contains a reversal for the $i$th component,
(2) the subrun $\pair{\alocation_0}{\vect{x}_0} \step{*} \pair{\alocation_{2K}}{\vect{x}_{2K}}$
satisfies PB$_{\sigma}$, (3) $i \in (\aset_1 \cup \cdots \cup \aset_K)$ 
and 
(4) for $j \in \interval{1}{n}$, 
$\vect{x}_{2K+2}(j) < \vect{x}_{2K+1}(j)$ implies $j \in (\aset_1 \cup \cdots \cup \aset_K)$.
Based on Lemmas~\ref{lemma-reversalbounded-unboundedness} and~\ref{lemma-simul},
a characterization for reversal-boundedness can be also defined.

\subsection{A first relaxation}
\label{section-first-relaxation} 

Below, we relax the satisfaction of the property $\mathcal{P}$ 
by allowing negative component values in a controlled way. 
A pseudo-run of the form
$\pair{\alocation_0}{\vect{x_0}} \step{\apath_0'} \pair{\alocation_1}{\vect{x}_1} \step{\apath_1} 
\pair{\alocation_2}{\vect{x}_2} 
             \step{\apath_1'} \pair{\alocation_3}{\vect{x}_3}  
             \cdots 
             \step{\apath_{K-1}'} \pair{\alocation_{2K-1}}{\vect{x}_{2K-1}} \step{\apath_K} 
             \pair{\alocation_{2K}}{\vect{x}_{2K}}$
\defstyle{weakly satisfies $\mathcal{P}$} $\equivdef$ 
it satisfies (P0), (P1), (P2) and (P3) defined as follows:
  for $j \in \interval{1}{n}$, 
every pseudo-configuration $\vect{x}$ such that $\vect{x}(j) < 0$ occurs after some 
$\vect{x}_{2l}$ for which  $\vect{x}_{2l}(j) - \vect{x}_{2l-1}(j) > 0$. 
If the run $\arun$ satisfies $\mathcal{P}$, then viewed as a pseudo-run, it also weakly
satisfies $\mathcal{P}$. Lemma~\ref{lemma-pseudorun-run} below 
states that the existence of pseudo-runs weakly satisfying  $\mathcal{P}$ is equivalent to
the existence of runs satisfying  $\mathcal{P}$ and their length can be compared. 
Later, we shall 
use the witness pseudo-run characterization.

\begin{lemma} \label{lemma-pseudorun-run}
Let $\apseudorun$ be a pseudo-run of length $L$ weakly satisfying $\mathcal{P}$ 
(of length $K$).
Then, there is a run $\arun$ satisfying $\mathcal{P}$ of length at most
$((L \times \pic{\avass})^{K} \times (1 + K^2 \times L \times \pic{\avass}) + L$. 
\end{lemma} 
The principle of the proof of Lemma~\ref{lemma-pseudorun-run} (and part of the proof
of Lemma~\ref{lemma-simul})
is identical to the idea of the proof of the following property of the coverability graph
$CG(\avass,\pair{\alocation_0}{\vect{x_0}})$ (see e.g., details
in~\cite{Reutenauer90}). For every extended configuration $\pair{\alocation}{\vect{y'}} \in \locations \times
(\Nat \cup \set{\infty})^n$  in  
           $CG(\avass,\pair{\alocation_0}{\vect{x_0}})$ and bound $B \in \Nat$,
there is a run $\pair{\alocation_0}{\vect{x_0}} \step{*} \pair{\alocation}{\vect{y}}$
      in $\avass$
      such that for $i \in \interval{1}{n}$, if $\vect{y'}(i) = \infty$ then 
      $\vect{y}(i) \geq B$ otherwise
      $\vect{y}(i) = \vect{y'}(i)$. 
In the proof of Lemma~\ref{lemma-pseudorun-run}, the paths $\apath_i$'s are repeated
hierarchically in order to eliminate negative values.  
Additionally, if 
$\apseudorun$ is a pseudo-run  of length $L$
weakly satisfying $\mathcal{P}$ and $L$ is at most doubly exponential in $ N = 
\length{\avass} + 
\length{\pair{\alocation_0}{\vect{x}_0}} + K + \absmaximum{\mathcal{P}}$,
then there is a run satisfying $\mathcal{P}$ and starting in $\vect{x}_0$ that is also
of length at most doubly exponential in  $N$.

\section{\expspace \ Upper Bound}
\label{section-expspace} 

In  this section, we deal with VAS only and we consider a current VAS $\avas$ of dimension $n$
(see Lemma~\ref{lemma-reduction}). 
W.l.o.g., we can assume that $n > 1$, 
otherwise it is easy to show that the generalized unboundedness problem restricted to VAS of dimension
1 can be solved in polynomial space.  Moreover, we  assume that $\pic{\avas} \geq 1$.

\subsection{Approximating generalized unboundedness properties}
\label{section-approximations}
Generalized unboundedness properties apply on runs but as it will be shown below, it would be more convenient
to relax the conditions to pseudo-runs. A first step has been done in Section~\ref{section-first-relaxation};
we shall push further the idea in order to adapt Rackoff's proof. 
Let $\apseudorun = \vect{x}_{0} \step{\apath_{0}'} \vect{x}_{1} \step{\apath_1} \vect{x}_{2}
              \cdots \vect{x}_{2K-1} \step{\apath_K} \vect{x}_{2K}$
be a pseudo-run weakly satisfying $\mathcal{P} = \tuple{\ainterval_1}{\ainterval_K}$. 
We suppose that $\apseudorun$ is induced by the path  
$\atransition_1 \ldots \atransition_k$ with $\apseudorun = \vect{u}_{0} \cdots \vect{u}_k$ and $\amap:
\interval{0}{2K} \rightarrow \interval{0}{k}$ is the map such that
$\vect{x}_i = \vect{u}_{\amap(i)}$ ($\amap(0) = 0$, $\amap(2K)=k$). 
For each position $j \in \interval{0}{\amap(2K-2)}$ along $\apseudorun$, there is a maximal $l_j \in \interval{1}{K}$ (with respect to standard ordering on
$\Nat$) and ${\rm INCR}_j \subseteq \interval{1}{n}$ such that
 $\amap(2l_j-2)  \leq j$ and
${\rm INCR}_j = \set{i \in \interval{1}{n}: \exists \ l' \in \interval{1}{l_j-1} 
\ {\rm such \ that} \ \vect{x_{2l'-1}}(i) < \vect{x_{2l'}}(i)}$.
In the induction proof of Lemma~\ref{lemma-main-induction}, we will need to check properties on suffixes of pseudo-runs and it will be useful
to approximate $\mathcal{P}$ with respect to some suffix  $\tuple{\ainterval_{l_j}}{\ainterval_K}$ and 
to some set of components
${\rm INCR}_j$. Indeed, the suffix $\vect{u}_{l_j} \cdots \vect{u}_k$ weakly satisfies $\tuple{\ainterval_{l_j}}{\ainterval_K}$ assuming that
we know how to increment strictly the components from ${\rm INCR}_j$. Moreover, like the notion of $i$-$B$-boundedness from~\cite{Rackoff78},
we would like to enforce that for each component $j$ from a given set $I$ and for each pseudo-configuration $\vect{y}$ along the pseudo-run
satisfying the approximation property, either $\vect{y}(j)$ belongs to $\interval{0}{B-1}$ or
the prefix pseudo-run terminating on $\vect{y}$ has the ability to increase arbitrarily the value $\vect{y}(j)$ (this will correspond to
condition (P2$'$) below). So, we are now in position to define the approximation property 
$\mathcal{A}[\mathcal{P},l, {\rm INCR},I,B]$. 
Given a generalized unboundedness property $\mathcal{P}$ of 
length $K$, $l \in \interval{1}{K}$, ${\rm INCR} \subseteq \interval{1}{n}$,
$I \subseteq \interval{1}{n}$ and $B \geq 0$,
a pseudo-run of the form below
$$\vect{y}_{2l-2} \step{\apath_{l-1}'} \vect{y}_{2l-1} \step{\apath_l} \vect{y}_{2l}
              \cdots 
             \step{\apath_{K-1}'} \vect{y}_{2K-1} \step{\apath_K} \vect{y}_{2K}$$
satisfies the approximation property
$\mathcal{A}[\mathcal{P},l,{\rm INCR}, I,B]$ (also abbreviated by $\mathcal{A}$) $\equivdef$ the conditions below are verified:
\begin{description}
\itemsep 0 cm
\item[(P1$'$)] For  $l' \in \interval{l}{K}$ and   $j \in \interval{1}{n}$, 
 we have $\vect{y}_{2l'}(j) - \vect{y}_{2l'-1}(j) \in \ainterval_{l'}(j)$.
\item[(P2$'$)] For  $l' \in \interval{l}{K}$ and   $j \in \interval{1}{n}$, 
 if $\vect{y}_{2l'}(j) - \vect{y}_{2l'-1}(j) < 0$, then  (there is
             $l \leq l'' < l'$ such that  $\vect{x}_{2l''}(j) - \vect{x}_{2l''-1}(j) > 0$
or $j \in {\rm INCR}$). 
\item[(P3$'$)]  For every pseudo-configuration $\vect{x}$ in $\arun$ occurring between
            $\vect{y}_{2l'}$ and strictly before $\vect{y}_{2l'+2}$ 
            with $l' \geq l-1$, $\vect{x}(J) \in \interval{0}{B-1}^J$ with 
             $J = I \setminus ({\rm INCR} \cup \set{j: \exists \ l \leq l'' \leq l', \
                  \vect{x}_{2l''}(j) - \vect{x}_{2l''-1}(j) > 0})$.
\end{description}
Condition (P3$'$) reflects the intuition that only the values from components in $J$ require to be controlled.
We also write $\mathcal{A}[\mathcal{P},l, {\rm INCR},I,+ \infty]$  
to denote the property obtained from $\mathcal{A}[\mathcal{P},l, {\rm INCR}, I,B]$
by replacing  $\interval{0}{B-1}^J$  by $\Nat^J$ in (P3$'$). 
Observe that a pseudo-run satisfies  
$\mathcal{A}[\mathcal{P},1,\emptyset, \interval{1}{n},+ \infty]$ iff 
it weakly satisfies
$\mathcal{P}$. 
The property $\mathcal{A}[\mathcal{P},l,  {\rm INCR},I,+ \infty]$ is exactly the condition we need in the proof
of Lemma~\ref{lemma-main-induction} below thanks to the property stated below. 
\begin{lemma} \label{lemma-repeat}
If the pseudo-run $\vect{y}_{2l-2} \step{\apath_{l-1}'} \vect{y}_{2l-1} \step{\apath_l} \vect{y}_{2l}
              \cdots  \step{\apath_K} \vect{y}_{2K}$
satisfies the approximation property $\mathcal{A}[\mathcal{P},l, {\rm INCR},I,+ \infty]$, then 
 $\pair{\apath_{l-1}' (\apath_l)^{n_l} \apath_{l}' (\apath_{l+1})^{n_{l+1}} \cdots 
(\apath_{K})^{n_{K}}}{\vect{y}_{2l-2}}
$
also satisfies it, for  $n_l, \ldots, n_K \geq 1$. 
\end{lemma}
A similar statement 
 does not hold for  pseudo-runs satisfying $\mathcal{A}$ 
(values for components in $J$ might become out of
$\interval{0}{B-1}$) and for runs satisfying 
$\mathcal{P}$ (component values might become negative).
Property  $\mathcal{A}[\mathcal{P},l, {\rm INCR}, I,B]$ can be viewed as a collection
of \emph{local} path increasing formulae in the sense of~\cite{Faouzi&Habermehl09}. 
\subsection{Bounding the length of pseudo-runs}
Let us briefly recall the structure of Rackoff's
proof to show that the boundedness problem for VAS is in \expspace. 
A witness run for unboundedness is of the form
$\arun = \vect{x}_{0} \step{*} \vect{y} \step{+} \vect{y}'$ with $\vect{y} \prec \vect{y}'$. 
In~\cite{Rackoff78}, it is shown that $\arun$ can be of length at most
doubly exponential.
In order to get the \expspace \ upper bound, Savitch's theorem is used.
Rackoff's proof to establish the small run property goes as follows.
First, a technical lemma shows that if there is some $i$-$B$-bounded pseudo-run (instance of
 the approximation property 
$\mathcal{A}$), then there is
one of length at most $B^{\length{\avas}^C}$ for some constant $C$. The proof essentially shows that
existence of such a pseudo-run amounts to solving an inequation system and by using~\cite{Borosh&Treybig76},
small solutions exist, whence the obtention of a short  $i$-$B$-bounded pseudo-run. 
The idea of using small solutions of inequation system to solve problems on counter systems
dates back from~\cite{Rackoff78,Gurari&Ibarra78} and nowadays, 
\iflong
\else this is a standard proof technique.
\fi
This proof can be extended to numerous properties on pseudo-runs for which intermediate counter value differences
can be expressed in Presburger arithmetic as done in~\cite{Yen92,Faouzi&Habermehl09}. 
Then, a proof by induction on the dimension is performed by using this very technical lemma and
the ability to repeat sequences of transitions; the proof can be extended when the first intermediate configuration is less or equal to the
last configuration of the sequence (leading to the concept of increasing path formula 
in~\cite{Faouzi&Habermehl09}).
 This condition allows to perform the induction
on the dimension with  a unique increasing formula. 
Unfortunately, generalized unboundedness properties are not increasing in the sense 
of~\cite{Faouzi&Habermehl09}
and therefore Rackoff's proof requires to be extended (but the main ingredients remain). 
The generalization of the
technical lemma is presented below; it is not surprising since generalized
unboundedness properties are  Presburger-definable properties. However, not only we 
need to refine the expression $B^{\length{\avas}^C}$ in terms of various 
parameters 
 (length of $\mathcal{P}$,
$\absmaximum{\mathcal{P}}$,
$n$, $\absmaximum{\avas}$)  
in order to get the final \expspace \ upper bound
(or the \pspace \ upper bound with fixed dimension), but also we have to check that the new ingredients
in  the definition of   $\mathcal{A}$ do not prevent us from 
extending~\cite[Lemma 4.4]{Rackoff78}. Finally, it is important to specify the length of small
pseudo-runs with respect to parameters from $\mathcal{P}$. 

\begin{lemma} \label{lemma-simple-loops}
Let $\avas$ be a VAS of dimension $n \geq 2$,
$\mathcal{P}$ be a generalized unboundedness property of length $K$,
$l \in \interval{1}{K}$, $B \geq 2$, 
$I, {\rm INCR} \subseteq \interval{1}{n}$ and 
$\apseudorun$ be a pseudo-run
satisfying  $\mathcal{A}[\mathcal{P},l, {\rm INCR},I,B]$. Then, there exists
a pseudo-run starting by the same pseudo-configuration, satisfying $\mathcal{A}[\mathcal{P},l, {\rm INCR},I,B]$
and of length at most $(1+K) \times (\absmaximum{\avas} \times \absmaximum{\mathcal{P}} \times
B)^{n^{\aconstant_1}}$ for some constant
$\aconstant_1$ independent of $K$, $\absmaximum{\mathcal{P}}$, $\absmaximum{\avas}$, $B$ and $n$.
\end{lemma}

The length expression in Lemma~\ref{lemma-simple-loops} can be certainly refined in terms
of $\card{{\rm INCR}}$, $\card{I}$ and $l$ but these values are anyhow bounded by
$n$ and $K$ respectively, which is used in  Lemma~\ref{lemma-simple-loops}.
For $i \in \interval{0}{n}$, 
\[g(i) \egdef
       \begin{cases}
       (2 \mu)^{n^{\aconstant_1}} \ {\rm with} \ \mu = (1 + K) \times  \absmaximum{\avas} \times \absmaximum{\mathcal{P}} & \text{if $i = 0$}, \\
       \big(2 \mu 
(\pic{\avas} \times g(i-1))\big)^{n^{\aconstant_1}} + g(i-1) & \text{if $i > 0$}.  
       \end{cases}
       \]
Lemma~\ref{lemma-main-induction} below is an extension of~\cite[Lemmas 4.6 \& 4.7]{Rackoff78}, see 
also~\cite[Lemma 7]{Faouzi&Habermehl09}. 
\begin{lemma} \label{lemma-main-induction} 
 Let $I, {\rm INCR} \subseteq \interval{1}{n}$, $l \in \interval{1}{K}$ and 
$\apseudorun$ be a pseudo-run  satisfying $\mathcal{A}[\mathcal{P},l, {\rm INCR},I, + \infty]$.
Then, there exists a pseudo-run $\apseudorun'$ starting from the same pseudo-configuration,  satisfying
the property $\mathcal{A}[\mathcal{P},l, {\rm INCR},I, + \infty]$
  and of length at most $g(\card{I})$.
\end{lemma}

In the induction step, we need to take advantage simultaneously of  the pigeonhole principle,
the induction hypothesis and Lemma~\ref{lemma-simple-loops}.

\begin{proof} Let 
$\apseudorun = \vect{x}_{2l-2} \step{\apath_{l-1}'} \vect{x}_{2l-1} \step{\apath_l} \vect{x}_{2l}
              \cdots 
             \step{\apath_{K-1}'} \vect{x}_{2K-1} \step{\apath_K} \vect{x}_{2K}$ be a pseudo-run
satisfying the property $\mathcal{A}[\mathcal{P},l, {\rm INCR},I, + \infty]$.
We suppose that $\apseudorun$ is induced by the path  
$\atransition_1 \ldots \atransition_k$ with $\apseudorun = \vect{u}_{0} \cdots \vect{u}_k$ and $\amap:
\interval{2l-2}{2K} \rightarrow \interval{0}{k}$ is the map such that
$\vect{x}_i = \vect{u}_{\amap(i)}$ ($\amap(2l-2) = 0$, $\amap(2K)=k$). 

The proof is by induction on $i = \card{I}$. If $i = 0$, then we apply Lemma~\ref{lemma-simple-loops} with 
$B = 2$ 
and we obtain a pseudo-run satisfying 
$\mathcal{A}[\mathcal{P},l, {\rm INCR},I, + \infty]$
leading to the bound 
$
(\mu \times 2)^{n^{\aconstant_1}}
$.

Now suppose $\card{I} = i+1$ and $J = (I \setminus {\rm INCR})$. 
We pose $B = \pic{\avas} \times g(i)$. 
We recall that $\avas$ is the current VAS with $n \geq 2$. 
We perform a case analysis depending where in  $\apseudorun$ 
a value from a component in $J$ is strictly greater than $B-1$ (if any). \\
{\bf Case 1:} Every configuration $\vect{z}$ in $\apseudorun$ satisfies $\vect{z}(J) \in
 \interval{0}{B-1}^J$, i.e., 
$\apseudorun$ satisfies $\mathcal{A}[\mathcal{P},l, {\rm INCR},I,B]$. \\
Obviously, the case $J = \emptyset$  is captured here. 
By Lemma~\ref{lemma-simple-loops}, there is a pseudo-run $\apseudorun'$ starting at $\vect{x}_{2l-2}$
satisfying   $\mathcal{A}[\mathcal{P},l, {\rm INCR},I,B]$ of length at 
most $(1+K) \times (\absmaximum{\avas}  \times
\absmaximum{\mathcal{P}} \times B)^{n^{\aconstant_1}}$, which is bounded by 
$\big(\mu \times 
(\pic{\avas} \times g(i))\big)^{n^{\aconstant_1}}$. \\
{\bf Case 2:} A value for some component in $J$ is strictly greater than $B-1$ for the first
time within the path  $\apath'_{D}$ for some  $D \in \interval{l-1}{K-1}$. 
Let $\alpha$ be the minimal position such that $\vect{u}_{\alpha+1}(J) \not \in \interval{0}{B-1}^J$  
and $\alpha+1 \in \interval{\amap(2D)+1}{\amap(2D+1)}$, say
$\vect{u}_{\alpha+1}(i_0) \geq B$ for some $i_0 \in J$. 
The pseudo-run $\apseudorun$ can be decomposed as follows with 
$\apath_{D}' = \apath_{D}^1 \atransition_{\alpha+1} \apath_{D}^2$ (${\rm INCR}'$ is defined few lines below):
$$
\underbrace{\vect{x}_{2l-2} \step{\apath_{l-1}'} \vect{x}_{2l-1} 
              \cdots \vect{x}_{2D}}_{\apseudorun_1} 
               = \underbrace{\vect{x}_{2D} \step{\apath_{D}^1} \vect{u}_{\alpha}}_{\apseudorun_2} 
              \step{\atransition_{\alpha+1}}
             \overbrace{
              \underbrace{
              \vect{u}_{\alpha+1}  \step{\apath_{D}^2} \vect{x}_{2D+1}
              \cdots
              \vect{x}_{2K-1} \step{\apath_K} \vect{x}_{2K}}_{\apseudorun_3}}^{{\rm satisfies} \
              \mathcal{A}[\mathcal{P},D+1,{\rm INCR}', (I \setminus \set{i_0}),+\infty]}
$$
We shall construct a pseudo-run of the form $\apseudorun'_1 \apseudorun'_2 \apseudorun'_3$
such that each $\apseudorun'_j$ is obtained by  shortening $\apseudorun_j$ and
the length of $\apseudorun'_1$ [resp. $\apseudorun'_2$, $\apseudorun'_3$]
is bounded by $(\mu \times B)^{n^{\aconstant_1}} + 1$ [resp. $B^{i+1} + 1$,  $g(i) + 1$].

\begin{itemize}
\itemsep 0 cm 
\item If $D > l-1$, then we introduce
$\mathcal{P}^{\star} = \tuple{\ainterval_l'}{\ainterval_D'}$ with
for $l'' \in \interval{l}{D}$ and $j \in \interval{1}{n}$, 
if $\vect{x}_{2l''}(j) - \vect{x}_{2l''-1}(j) > 0$ then 
$\ainterval_{l''}'(j) = \ainterval_{l''}(j) \cap [1,+\infty[$,
otherwise $\ainterval_{l''}'(j) = \ainterval_{l''}(j)$. 
The construction of $\mathcal{P}^{\star}$ allows us to preserve the set of components
in $\interval{l}{D}$ whose values can be arbitrarily increased. 
By  Lemma~\ref{lemma-simple-loops}, there is a pseudo-run 
$\apseudorun_1' = \pair{\atransition_1^1 \cdots 
\atransition_{\beta_1}^1}{\vect{x}_{2l-2}}$  satisfying 
$\mathcal{A}[\mathcal{P}^{\star},1, {\rm INCR}, I, B]$ 
such that $\beta_1 \leq  (\mu \times B)^{n^{\aconstant_1}}$.
Indeed, $\absmaximum{\mathcal{P}^{\star}} \leq \absmaximum{\mathcal{P}}$
and the length of $\mathcal{P}^{\star}$ is obviously bounded by $K$. 
Say $\apseudorun_1' = \vect{y}_{2l-2} \step{*} \vect{y}_{2l-1} \step{*} \vect{y}_{2l}
              \cdots 
             \step{*} \vect{y}_{2D-1} \step{*} \vect{y}_{2D}$.
Suppose that $\apseudorun_1' = \vect{u}_{0}^1 \cdots \vect{u}_{\beta_1}^1$ and $\amap_1:
\interval{2l-2}{2D}
\rightarrow \interval{0}{\beta_1}$ is the map such that
$\vect{y}_i = \vect{u}_{\amap_1(i)}^1$ ($\amap_1(2l-2) = 0$, $\amap_1(2D)=\beta_1$).
If $D = l-1$, then $\apseudorun_1 = \pair{\atransition_1 \cdots \atransition_{\alpha}}{\vect{x}_{2l-2}}$
with an analogous decomposition in terms of $\vect{y}_i$'s. 
We have  
$\set{j: 
\vect{y}_{2l'-1}(j) < \vect{y}_{2l'}(j), \ l' \in \interval{l}{D}} = \set{j: 
\vect{x}_{2l'-1}(j) < \vect{x}_{2l'}(j), \ l' \in \interval{l}{D}}$ ($\egdef \asetter$) --partly
by construction of  $\mathcal{P}^{\star}$.
\item Now, by the piegonhole principle, there is a pseudo-run 
$\apseudorun'_2 = \pair{\atransition_1^2 \cdots 
\atransition_{\beta_2}^2}{\vect{y}_{2D}}$ such that  
 $\vect{u}'_{\alpha} = \vect{y}_{2D} + \atransition_1^2 + \cdots +  \atransition_{\beta_2}^2$,
$\vect{u}'_{\alpha}(J) = \vect{u}_{\alpha}(J)$ and $\beta_2 < B^{\card{J}} \leq B^{i+1}$. 
We pose $\vect{u}'_{\alpha+1} = \vect{u}'_{\alpha} + \atransition_{\alpha+1}$. 
\item Finally, observe that $\pair{\atransition_{\alpha+2} \cdots \atransition_{k}}{\vect{u}'_{\alpha+1}}$
satisfies $\mathcal{A}[\mathcal{P},D+1,{\rm INCR'},(I \setminus \set{i_0}), + \infty]$ with 
${\rm INCR}' \egdef {\rm INCR} \cup \asetter$.
By the induction hypothesis, there is a pseudo-run $\apseudorun_3' = 
\pair{\atransition_1^3 \cdots \atransition_{\beta_3}^3}{\vect{u}'_{\alpha+1}}$
satisfying  
$\mathcal{A}[\mathcal{P},D+1,{\rm INCR'},(I \setminus \set{i_0}), + \infty]$
and such that
$\beta_3 \leq g(i)$.
Because $\vect{u}_{\alpha+1}'(i_0) \geq \pic{\avas}  \times g(i)$, $\apseudorun_3'$ also satisfies
 $\mathcal{A}[\mathcal{P},D+1,{\rm INCR'},I, + \infty]$.
\end{itemize}
Glueing the previous transitions, the pseudo-run 
$\pair{
\atransition_1^1 \cdots \atransition_{\beta_1}^1
\atransition_1^2 \cdots \atransition_{\beta_2}^2
\atransition_{\alpha + 1} 
\atransition_1^3 \cdots \atransition_{\beta_3}^3}
{\vect{x}_{2l-2}}$
satisfies   the approximation property 
$\mathcal{A}[\mathcal{P},l,{\rm INCR},I, + \infty]$.
 and its length is bounded by
$
 (\mu \times B)^{n^{\aconstant_1}} + B^{i+1} +  g(i)
$.\\
{\bf Case 3:}  A value for some component in $J$ is strictly greater than $B-1$ for the first
time within the path  $\apath_{D}$ for some  $D \in \interval{l}{K}$. \\ 
The pseudo-run $\apseudorun$ can be written as follows with 
$\apath_D = \apath_D^1 \apath_D^2$ and $\apath_D^1 \neq \varepsilon$
$$
\vect{x}_{2l-2} \step{\apath_{l-1}'} \vect{x}_{2l-1} 
              \cdots \vect{x}_{2D-1} 
               \step{\apath_{D}^1} \vect{u}_{\alpha+1} \step{\apath_{D}^2} \vect{x}_{2D}
              \cdots
              \vect{x}_{2K-1} \step{\apath_K} \vect{x}_{2K}
$$
By Lemma~\ref{lemma-repeat}, the pseudo-run
$ \apseudorun'
\pair{\apath_{l-1}' \apath_l \cdots \apath'_{D-1} (\apath_D)^2 \apath'_{D} \cdots 
\apath_{K}}{\vect{x}_{2l-2}}
$
also satisfies the approximation property $\mathcal{A}[\mathcal{P},l,{\rm INCR},I, + \infty]$ and can be written as
$
\vect{x}_{2l-2} \step{\apath_{l-1}'} \vect{x}_{2l-1} 
              \cdots \vect{x}_{2D-2} \step{\apath_{D-1}' \apath_D} \vect{x}_{2D} = \vect{z}_{2D-1}
              \step{\apath_D} \vect{z}_{2D} \step{\apath'_{D+1}} \cdots
              \vect{z}_{2K-1} \step{\apath_K} \vect{z}_{2K}
$.
We are therefore back to Case 2.
\end{proof}

\iflong
\else
Now, we are seeking to bound $g(n)$. 
\fi 

\begin{lemma} \label{lemma-length}
If $\apseudorun$ is a pseudo-run weakly satisfying $\mathcal{P}$,
then there is a $\apseudorun'$ starting from the same pseudo-configuration, weakly satisfying
 $\mathcal{P}$ and of length at most  
$(\mu \times 2 \times \pic{\avas})^{n^{(2n+1) \aconstant}}$ for some
 $\aconstant > 1$ with $\mu = (1 + K) \times  \absmaximum{\avas} \times 
\absmaximum{\mathcal{P}}$.
\end{lemma}

\begin{proof}
Let us bound $g(n)$. 
By Lemma~\ref{lemma-main-induction},  for some constant 
$\aconstant_2 > \aconstant_1$ (for instance $\aconstant_2 = \aconstant_1 + 1$),
we have 
\[g(i) \leq
       \begin{cases}
       (2 \mu)^{n^{\aconstant_2}} & \text{if $i = 0$}, \\
       \big(2 \mu (\pic{\avas} \times g(i-1))\big)^{n^{\aconstant_2}} & \text{if $i > 0$}.  
       \end{cases}
       \]
By induction on $i$, we can show that
 $g(i) \leq (\nu^{i+1})^{n^{(2i+1) 
\aconstant_2}}$ with $\nu = 2 \mu \times \pic{\avas}$. For $i = 0$ this is obvious . Otherwise
$$
g(i+1) \leq  \big(2 \mu  \times \pic{\avas} \times g(i)\big)^{n^{\aconstant_2}}
\leq 
(\nu
 (\nu^{i+1})^{n^{(2i+1) 
\aconstant_2}})^{n^{\aconstant_2}} \leq \ldots
$$
$$
\leq
((\nu^{i+2})^{n^{(2i+1) 
\aconstant_2}})^{n^{\aconstant_2}}
\leq
(\nu^{i+2})^{n^{(2i+2) \aconstant_2}}
<
(\nu^{i+2})^{n^{(2i+3) \aconstant_2}}
$$
Hence, $g(n) \leq (\nu^{n+1})^{n^{(2n+1) \aconstant_2}}$.
As soon as $n \geq 2$, 
there is a constant $\aconstant$ s.t. 
$g(n) \leq (2 \mu  \times \pic{\avas})^{n^{(2n+1) \aconstant}}$.
\end{proof}

Let us conclude the section by the main result of the paper. 

\begin{theorem} \label{theorem-gene} 
(I) The generalized unboundedness problem for VASS is \expspace-com\-plete.
(II)  For each $n \geq 1$, the generalized unboundedness problem  
restricted to VASS of dimension at most $n$ is in \pspace.
\end{theorem}%

%
\section{Other Applications}%

In this section, we draw conclusions from  Theorem~\ref{theorem-gene}.
First, as a by-product of Theorem~\ref{theorem-gene} and using the reductions from Section~\ref{section-expressive-power}, 
we can easily regain the exponential-space bound mentioned below.

\begin{corollary} \label{corollary-main}
The regularity detection problem
and the strong promptness detection problem are in \expspace.  
The simultaneous unboundedness problem is \expspace-complete. 
For each fixed $n \geq 1$, their restriction to VASS of dimension at most $n$ are in \pspace.
\end{corollary}

The complexity upper bound for regularity detection problem has been left open 
in~\cite{Faouzi&Habermehl09}. 
Decidability of the strong promptness detection problem is 
established in~\cite{Valk&Jantzen85}. 
The \expspace \ upper bound has been already stated in~\cite{Yen92,Faouzi&Habermehl09}. 
We cannot rely on~\cite{Yen92} because of the 
flaw in ~\cite[Lemma 7.7]{Yen92}.
Condition 4. in~\cite[page 13]{Faouzi&Habermehl09} does not characterize strong promptness
(but only promptness) as 
shown in Section~\ref{section-expressive-power}.
Finally, increasing path 
formulae from~\cite{Faouzi&Habermehl09}
cannot characterize strong promptness detection unlike 
generalized unboundedness properties. Therefore, we also believe that the upper bound
for strong promptness detection is new. 
Below, we state how the previous results allow us to characterize the computational complexity
of reversal-boundedness detection problem for VASS and its variant with weak reversal-boundedness. 

\begin{theorem} \label{theorem-RB} \
\iflong
\else
(I) Reversal-boundedness detection problem for VASS is  \expspace-complete.
(II) For each fixed $n \geq 1$,  its restriction to VASS of dimension at most $n$
is in \pspace.
(III) Properties (I) and (II) also hold true for weak reversal-boundedness detection problem.
\fi 
\end{theorem}

By Theorem~\ref{theorem-RB}(I), once an initialized 
VASS is shown to be reversal-bounded, one can compute 
effectively semilinear sets corresponding to reachability sets, 
see recent developments in~\cite{To10}. The size of the representation of such sets is at least polynomial
in  the maximal number of reversals.
However, we know that an initialized VASS can be bounded but still the cardinality of
its reachability set may be nonprimitive
recursive, see e.g.~\cite{Valk&VidalNaquet81}. A similar phenomenon occurs 
with reversal-boundedness, as briefly explained below. 
In case of reversal-boundedness, the maximal reversal can be 
nonprimitive recursive in the size of the initialized VASS in the 
\iflong
\else
worst-case.
\fi
Indeed,  
given $n \geq 0$, one can compute in time polynomial
in $n$ an initialized VASS $\pair{\avass_n}{\pair{\alocation_0}{\vect{x}_n}}$ that generates a 
finite reachability
set of cardinal $\mathcal{O}(A(n))$ for some nonprimitive recursive map $A(\cdot)$ similar to Ackermann
function, see e.g., the construction in~\cite{Jantzen86}. 
Moreover, $\pair{\avass_n}{\pair{\alocation_0}{\vect{x}_n}}$ 
can be shown  to admit only finite runs, see details 
in~\cite{Jantzen86}. It is then easy to compute a variant VASS $\avass_n'$ by adding a component and such that
each transition of $\avass_n$ is replaced by itself followed by  incrementating  the new component
and then   decrementing it (creating a reversal). 
Still $\avass_n'$ has no infinite computation, 
$\pair{\avass_n'}{\pair{\alocation_0}{\vect{x}_n'}}$ is reversal-bounded ($\vect{x}'_n$ restricted
to the components of $\avass_n$ is equal to $\vect{x}_n$) and its maximal reversal is in
$\mathcal{O}(A(n))$. 

\section{Concluding Remarks}
\label{section-conclusion}
We have proved the  \expspace-easiness of the generalized unboundedness 
problem  
(both the initialized VASS and the generalized unboundedness property are 
part of the inputs). For example, this allows us to provide the
optimal complexity upper bound for  the 
reversal-boundedness detection problems, place boundedness problem, strong 
promptness detection problem and regularity detection problem. 
Even though our proof technique is clearly tailored along the lines of~\cite{Rackoff78}, we had 
to provide a series of adaptations in order to get the final \expspace \ upper bound (and the \pspace \
upper bound for fixed dimension). In particular, we advocate the use of  witness pseudo-run
characterizations (instead of using runs) when there exist decision procedures using
coverability graphs.

Let us conclude by possible continuations.
\iflong
\else
Our \expspace \ proof can be obviously extended 
by  replacing intervals in properties by more complex
sets of integers or by adding new constraints between intermediate configurations.
\fi
The robustness of our proof technique still deserves to be determined. 
A challenging question is to determine the complexity of checking when a reachability set 
obtained by an initialized VASS is semilinear. 
Besides, various subclasses of VASS exist for which
decision problems are of lower complexity. 
For instance, in~\cite{MPraveen&Lodaya09}, 
the  boundedness problem is shown to be in \pspace \ for a class of VASS with  so-called 
bounded \defstyle{benefit depth}.
It is unclear for which subclasses of VASS,  the generalized unboundedness problem can be solved in polynomial
space too.%

{\em Acknowledgments:} I would like to thank Thomas Wahl (U. of Oxford) and anonymous referees 
for their suggestions and remarks about a preliminary version of this work. 

\bibliographystyle{eptcs}


\end{document}